\documentclass[a4paper,11pt,reqno]{amsart}
\usepackage{amsmath, amsthm, amssymb, mathrsfs, amsfonts}
\usepackage{array}
\usepackage[active]{srcltx}
\usepackage{graphicx}
\usepackage[margin=3cm]{geometry}
\usepackage[english]{babel}
\usepackage[colorlinks=true, allcolors=blue]{hyperref}
\usepackage{enumerate}
\usepackage{float}

\theoremstyle{definition}

\newtheorem{ejem}{Example}[section]
\newtheorem{ques}{Question}[section]
\theoremstyle{plain}
\newtheorem{teor}{Theorem}[section]

\newtheorem{cor}{Corollary}[section]

\def\diam{\mathop{ \rm diam}\nolimits}

\def\deg{\mathop{\rm deg}\nolimits}
\DeclareMathOperator{\wt}{wt}
\DeclareMathOperator{\supp}{supp}

\usepackage{tikz,xcolor,hyperref}

\definecolor{lime}{HTML}{A6CE39}
\DeclareRobustCommand{\orcidicon}{%
	\begin{tikzpicture}
	\draw[lime, fill=lime] (0,0)
	circle [radius=0.16]
	node[white] {{\fontfamily{qag}\selectfont \tiny ID}};
	\draw[white, fill=white] (-0.0625,0.095)
	circle [radius=0.007];
	\end{tikzpicture}
	\hspace{-2mm}
}

\foreach \x in {A, ..., Z}{%
	\expandafter\xdef\csname orcid\x\endcsname{\noexpand\href{https://orcid.org/\csname orcidauthor\x\endcsname}{\noexpand\orcidicon}}
}

\title{A graphical representation of binary linear codes}

\author[L. Delgado]{Lisbeth Danyeli Delgado Ordoñez\orcidC{}}
\address{Lisbeth D. Delgado Ordoñez, Departamento de Matem\'aticas, Universidad del Cauca}
\email{lddelgado@unicauca.edu.co}

\author[J.H. Castillo]{John H. Castillo\orcidB{}}
\address{John H. Castillo, Departamento de Matem\'aticas y Estad\'istica, Universidad de Nari\~no}
\email{jhcastillo@udenar.edu.co}

\author[A. Holgu\'in-Villa]{Alexander Holgu\'in-Villa\orcidA{}}
\address{Alexander Holgu\'in-Villa, Escuela de Matem\'aticas, Universidad Industrial de Santander}
\email{aholguin@uis.edu.co}


\begin{document}
 \noindent
  \renewcommand{\refname}{References}


\begin{abstract}
A binary $[n,k]$-linear code $\mathcal{C}$ is a $k$-dimensional subspace of $\mathbb{F}_2^n$. For $\boldsymbol{x}\in \mathbb{F}_2^n$, the set $\boldsymbol{x}+\mathcal{C}$ is a coset of $\mathcal{C}$. In this work we study a partial ordering on the set of cosets of a binary linear code
$\mathcal{C}$ of length $n$ and we construct a graph using the orphan structure of this code.
\end{abstract}

\maketitle

\section{Introduction}
The connection between Graph Theory and Coding Theory has been considered by several researchers. The adjacency matrix of a simple graph
with $n$ vertices, labeled by $v_1, v_2,\ldots , v_n$, is a square matrix  of order $n$, whose $(i, j)$-entry is equal to $1$ if the
vertices $v_i$ and $v_j$ are adjacent and is equal to zero otherwise, i.e., the adjacency matrix of a simple graph is a symmetric binary
matrix with diagonal zero which has make it suitable for constructing binary codes, see \cite{rouayheb2012graph,jungnickel1996codes,mallik2021graph,tonchev2002error} and the references quoted therein. Different types of
codes are obtained from this construction depending of the structure of the involved graph. On the other hand, the
work of M. Tanner \cite{tanner1981recursive} introduced the construction of bipartite graphs from the parity-check matrix of codes. Special focus has been dedicated in the application of Tanner graphs to the encoding and decoding algorithms for low-density parity-check codes (LDPC). Intensive work has been done
to analyze the properties of these graphs and to study the characteristics of the related codes, see \cite{etzion1999codes,kschischang2003codes,banihashemi2001tanner,halford2006codes}.

Let $\mathbb{F}_2$ be the field with $2$ elements. A \emph{binary $[n,k]$-linear code} is a $k$-dimensional subspace of $\mathbb{F}_2^n$.
An element of a binary linear code is called a \emph{codeword}. The \emph{Hamming distance}, $d(\boldsymbol{x},\boldsymbol{y})$, between two codewords $\boldsymbol{x}=(x_1,\ldots,x_n), \boldsymbol{y}=(y_1,\ldots,y_n)\in \mathcal{C}\subseteq\mathbb{F}_2^n$ is the number of entries where they differ,
or equivalently,

$$
d(\boldsymbol{x},\boldsymbol{y})=d(\boldsymbol{x},\boldsymbol{y}) =|\{i:x_i\neq y_i, ~1\leq i\leq n\}|.
$$

For $\boldsymbol{x}\in\mathbb{F}_2^n$, the \emph{Hamming weight} of $\boldsymbol{x}$ is $\wt(\boldsymbol{x})=d(\boldsymbol{x},\boldsymbol{0})$, i.e.,
$\wt(\boldsymbol{x})$ is the number of non-zero coordinates in $\boldsymbol{x}$. The \emph{minimum distance} $d(\mathcal{C})=d$ of a linear
code $\mathcal{C}$ is defined as the minimum weight among all non-zero codewords, thus we called it a binary $[n,k,d]$-linear
code. A \emph{generator matrix} for an $[n, k]$-linear code $\mathcal{C}$ is any $k\times n$ matrix $G$ whose rows form a basis for
$\mathcal{C}$. So the code $\mathcal{C}$ can be seen as

\begin{equation}\label{generator_matrix}
\mathcal{C}=\{ \boldsymbol{x}G: \boldsymbol{x}\in \mathbb{F}_2^k\}.
\end{equation}

Also, as a binary linear code is a subspace of a vector space, it is the kernel of some linear transformation. In particular,
there is an $(n-k)\times n$ matrix $H$, called a \emph{parity-check matrix} for the $[n, k]$-linear code $\mathcal{C}$, such that

\begin{equation}\label{check_matrix}
\mathcal{C} =\{\boldsymbol{x} \in \mathbb{F}_2^n:  H\boldsymbol{x}^T = \boldsymbol{0}\}.
\end{equation}

By using elementary row and column operations, we can bring the generator matrix $G$ into the \emph{standard form} $[I_k|A]$ where
$I_k$ denotes the $k\times k$ identity matrix and $A$ is a $k\times (n-k)$ matrix and, then we can obtain a parity-check
matrix for the code $\mathcal{C}$ in the standard form, see \cite[Thm. 1.2.1]{Huffman2003}.

\begin{teor}
Let $\mathcal{C}$ be an $[n, k]$-linear code over $\mathbb{F}_2$ with generator matrix $G=[I_k|A]$. Then $H=[A^t|I_{n-k}]$
is a parity-check matrix for $\mathcal{C}$.
\end{teor}

In this paper, we present a graphical representation (\emph{Hasse diagram}) of a binary linear code using the ideas given in
Section 11.7 of \cite{Huffman2003}. Our main interest is to identify properties of the graph constructed:
we demonstrate that the graph defined here is a connected bipartite graph which no contain triangles. In the next section, we recall some basic concepts and results from Coding Theory and Graph Theory and we give our definition of the graph obtained from binary linear code. Finally, in Section \ref{sec3} we present our main results.

\section{Preliminaries}

In general, there are two basic process to do with a binary $[n,k]$-linear code $\mathcal{C}$: encoding and decoding.
The first one consists in convert a vector of $\mathbb{F}_2^k$ and get a codeword, this process can be done using
\eqref{generator_matrix}. In the other hand, to decode a received word it is necessary, in a first stage, to identify
whether it is a codeword or not. It can be solve by mean of \eqref{check_matrix}. But when we are sending a word it
could be appear some errors in the received word, so we need to detect them and (if it is possible) correct them. These
two properties can be estimated using the minimum distance of $\mathcal{C}$, see \cite[Thms. 2.5.6 and 2.5.10]{Ling2004}.

For an $[n,k,d]$-binary linear code $\mathcal{C}$, we can devise an algorithm using a table with $2^{n-k}$ rather than
$2^n$ elements where one can find the nearest codeword by looking up one of these $2^{n-k}$ entries. This general decoding
algorithm for linear codes is called \emph{syndrome decoding}. Because our code $\mathcal{C}$ is an elementary abelian subgroup
of the additive group of $\mathbb{F}_2^n$, its distinct cosets $\boldsymbol{x} + \mathcal{C}$ partition $\mathbb{F}_2^n$ into
$2^{n-k}$ sets of size $2^k$. Two vectors $\boldsymbol{x}$ and $\boldsymbol{y}$ belong to the same coset if and only if
$\boldsymbol{y}-\boldsymbol{x} \in \mathcal{C}$. We denote by
$\mathfrak{cl}(\mathcal{C})=\{\boldsymbol{x} + \mathcal{C}: \boldsymbol{x}\in \mathbb{F}_2^n\}$  the set of cosets of the code
$\mathcal{C}$.

The table with the $2^{n-k}$ elements can be constructed as follows: Let $H$ be a parity-check matrix for $\mathcal{C}$,
that is a generator matrix for the set $\mathcal{C}^{\perp}=\{\boldsymbol{x}\in \mathbb{F}_2^n: \boldsymbol{x}\cdot \boldsymbol{c}=0, \text{ for all $\boldsymbol{c}\in \mathcal{C}$}\}$. The \emph{syndrome} of a word $\boldsymbol{y} \in \mathbb{F}_2^n$ with respect to $H$ is the vector
$S(\boldsymbol{y}) = H\boldsymbol{y}^T \in \mathbb{F}_2^{n-k}$. Since the syndrome of a codeword is $\boldsymbol{0}$, if
$\boldsymbol{y_1}, \boldsymbol{y_2}\in \mathbb{F}_2^n$ are in the same coset of $\mathcal{C}$, then $\boldsymbol{y_1} - \boldsymbol{y_2}=\boldsymbol{c}\in \mathcal{C}$.
Therefore $S(\boldsymbol{y_1})=H(\boldsymbol{y_2}+\boldsymbol{c})^T=H\boldsymbol{y_2}^T+H\boldsymbol{c}^T=S(\boldsymbol{y_2})$. Conversely,
if $S(\boldsymbol{y_1})=S(\boldsymbol{y_2})$, then $H(\boldsymbol{y_2}-\boldsymbol{y_1})^T=\boldsymbol{0}$ and so $\boldsymbol{y_2}-\boldsymbol{y_1}\in \mathcal{C}$,
i.e., we have a way to test whether the word belongs to the code. Moreover, there is a one-to-one correspondence between cosets of
$\mathcal{C}$ and syndromes. In fact, let $\mathcal{C}$ be a binary linear code and  assume the codeword $\boldsymbol{v}$
is transmitted and the word $\boldsymbol{w}$ is received, resulting in the \emph{error pattern} (or error string)
$\boldsymbol{e} = \boldsymbol{w} - \boldsymbol{v} \in \boldsymbol{w} + \mathcal{C}$. Then $\boldsymbol{w} - \boldsymbol{e} = \boldsymbol{v} \in \mathcal{C}$,
so, the error pattern $\boldsymbol{e}$ and the received word $\boldsymbol{w}$ are in the same coset.

Since error patterns of small weight are the most likely to occur, nearest neighbor decoding works for a linear code $\mathcal{C}$ in the following manner. Upon receiving the word $\boldsymbol{w}$, we choose a word $\boldsymbol{e}$ of least weight in the coset $\boldsymbol{w} + \mathcal{C}$ and conclude that $\boldsymbol{v} = \boldsymbol{w}-\boldsymbol{e}$ was the codeword transmitted. To do this it is necessary to construct the (Slepian) \emph{standard array}, for which we need to find the weight of a coset which is the smallest weight of a vector in the coset. Any vector of this smallest weight in the coset is called a \emph{coset leader}. The zero vector is the unique coset leader of the code $\mathcal{C}$. More generally, every coset of weight at most $t = \lfloor(d-1)/2\rfloor$ has a unique coset leader.

There is a natural partial ordering $\preceq$ on the vectors in $\mathbb{F}_2^n$, which is defined as follows: for
$\boldsymbol{x},\boldsymbol{y}\in \mathbb{F}_2^n$, $\boldsymbol{x}\preceq \boldsymbol{y}$ provided that $\supp(\boldsymbol{x}) \subseteq \supp(\boldsymbol{y})$,
where $\supp(\boldsymbol{c})$ for $\boldsymbol{c}\in \mathbb{F}_2^n$ denote the {\it support} of the vector $\boldsymbol{c}$, i.e.,
the set of coordinates where $\boldsymbol{c}$ is non-zero. If $\boldsymbol{x}\preceq \boldsymbol{y}$, we will also say that
$\boldsymbol{y}$ covers $\boldsymbol{x}$. We now use this partial order on $\mathbb{F}_2^n$ to define a partial order, also
denoted $\preceq$, on the set of cosets of a binary linear code $\mathcal{C}$ of length $n$. If $\mathcal{C}_1$ and
$\mathcal{C}_2$ are two cosets of $\mathcal{C}$, then $\mathcal{C}_1 \preceq \mathcal{C}_2$ provided there are coset
leaders $\boldsymbol{x}_1$ of $\mathcal{C}_1$ and $\boldsymbol{x}_2$ of $\mathcal{C}_2$ such that $\boldsymbol{x}_1 \preceq \boldsymbol{x}_2$.
As usual, $\mathcal{C}_1 \prec\mathcal{C}_2$ means that $\mathcal{C}_1 \preceq \mathcal{C}_2$ but $\mathcal{C}_1\neq \mathcal{C}_2$.
Under this partial ordering the set of cosets of $\mathcal{C}$ has a unique minimal element, the code $\mathcal{C}$
itself. Since, the set $\mathfrak{cl}(\mathcal{C})$ is a  partially ordered set, \emph{poset} for short, with respect to
$\preceq$, we can use a Hasse diagram to represent the partial order on $\mathfrak{cl}(\mathcal{C})$.

Let  $\boldsymbol{x}\in \mathbb{F}_2^n$ and $r$ be a non-negative real number. The \emph{sphere of radius  $r$ and center in $\boldsymbol{x}$}
is the set $S(\boldsymbol{x},r)=\{\boldsymbol{y}\in \mathbb{F}_2^n: d(\boldsymbol{x},\boldsymbol{y})\leq r\}.$ The \emph{covering radius} of a code $\mathcal{C}$, denoted by $\rho = \rho(\mathcal{C})$, is the smallest integer
$s$ such that $\mathbb{F}_2^n$ is the union of the spheres of radius $s$ centered at the codewords of $\mathcal{C}$.
Actually, it can be proved that $$\rho(\mathcal{C})=\max_{\boldsymbol{x}\in \mathbb{F}_2^n}\min_{\boldsymbol{c}\in\mathcal{C}}d(\boldsymbol{x},\boldsymbol{c}).$$

Also, by  \cite[Thm. 11.1.2]{Huffman2003}, $\rho(\mathcal{C})$ is the largest value of the minimum weight of all cosets of $\mathcal{C}$, i.e. $\rho(\mathcal{C})=\max\{\wt(\boldsymbol{x}+\mathcal{C}): \boldsymbol{x}\in \mathbb{F}_2^n\}.$

A \emph{graph} $\Gamma$ with $n$ vertices is a pair $(V,E)$ where $V=\{v_1, v_2,\ldots ,v_n\}$ is the set of vertices and
$E \subseteq V\times V$ is the set of edges. Given two vertices $v_i$ and $v_j$, if $v_iv_j\in E$, then $v_i$ and
$v_j$ are said to be \emph{adjacent} or that $v_i$ and $v_j$ are \emph{neighbors}. In this case, $v_i$ and $v_j$ are said to be
the end vertices of the edge $v_iv_j$. If $v_iv_j\notin E$, then $v_i$ and $v_j$ are \emph{non-adjacent}. Furthermore, if
an edge $e$ has a vertex $v_i$ as an end vertex, we say that $v_i$ is incident with $e$. We denote the \emph{set of neighbors} or \emph{neighborhood} of the vertex $v_i$ by $\mathcal{N}(v_i)$,
i.e., $\mathcal{N}(v_i)=\{x\in V: v_ix\in E\}$, and the cardinality  of this set is the \emph{degree} of the vertex $v_i$,
which is denoted by $\deg(v_i)$, in other words, the $\deg(v_i)$ is the number of edges incident with $v_i$.

Given a binary linear code $\mathcal{C}$, we denote by  $\Gamma(\mathcal{C})=(V_{\mathcal{C}},E_{\mathcal{C}})$ the  graph
constructed on this wise: the set of vertices  $V_{\mathcal{C}}=\mathfrak{cl}(\mathcal{C})$  and $\mathcal{C}_1\mathcal{C}_2\in E_{\mathcal{C}}$
if $\mathcal{C}_1\prec \mathcal{C}_2$ and $\wt(\mathcal{C}_1)=\wt(\mathcal{C}_2)-1$. If $\mathcal{C}_1\mathcal{C}_2$ is an edge of $\Gamma(\mathcal{C})$, $\mathcal{C}_1$ is called a \emph{child} of $\mathcal{C}_2$, and $\mathcal{C}_2$ is a \emph{parent} of $\mathcal{C}_1$. Actually, the graph $\Gamma(\mathcal{C})$ is known as the \emph{Hasse diagram} of the poset $(\mathfrak{cl}(\mathcal{C}),\preceq)$.  It is clear that $|V_{\mathcal{C}}|=2^{n-k}$.
The main interest of this manuscript is to study the properties of the graph $\Gamma(\mathcal{C})$.

\begin{ejem}
\label{ejemplorelaciones}
	Let $\mathcal{C}$ be a $[5, 2, 3]$ binary linear code with generator matrix
	\begin{center}	
		$G=\begin{pmatrix}
			1 & 1 & 1 & 1 & 0 \\
			0 & 0 & 1 & 1 & 1	
		\end{pmatrix}.$
	\end{center}
	The cosets of $\mathcal{C}$ are
	\begin{align*}
		\mathcal{C}_0 &= 00000 + \mathcal{C},    &   \mathcal{C}_4 &= 00010 + \mathcal{C}, \\
		\mathcal{C}_1 &= 10000 + \mathcal{C},    &   \mathcal{C}_5 &= 00001 + \mathcal{C},\\
		\mathcal{C}_2 &= 01000 + \mathcal{C},    &   \mathcal{C}_6 &= 10100 + \mathcal{C} = 01010 + \mathcal{C},\\
		\mathcal{C}_3 &= 00100 + \mathcal{C},    &   \mathcal{C}_7 &= 10010 + \mathcal{C} = 01100 + \mathcal{C}.
    \end{align*}

The visual representation of the graph (Hasse diagram) $\Gamma({\mathcal{C}})$ is given in the Figure \ref{fig:1}.
	\begin{figure}[H]
		\centering
\hspace*{2.5cm}\includegraphics[scale=0.7]{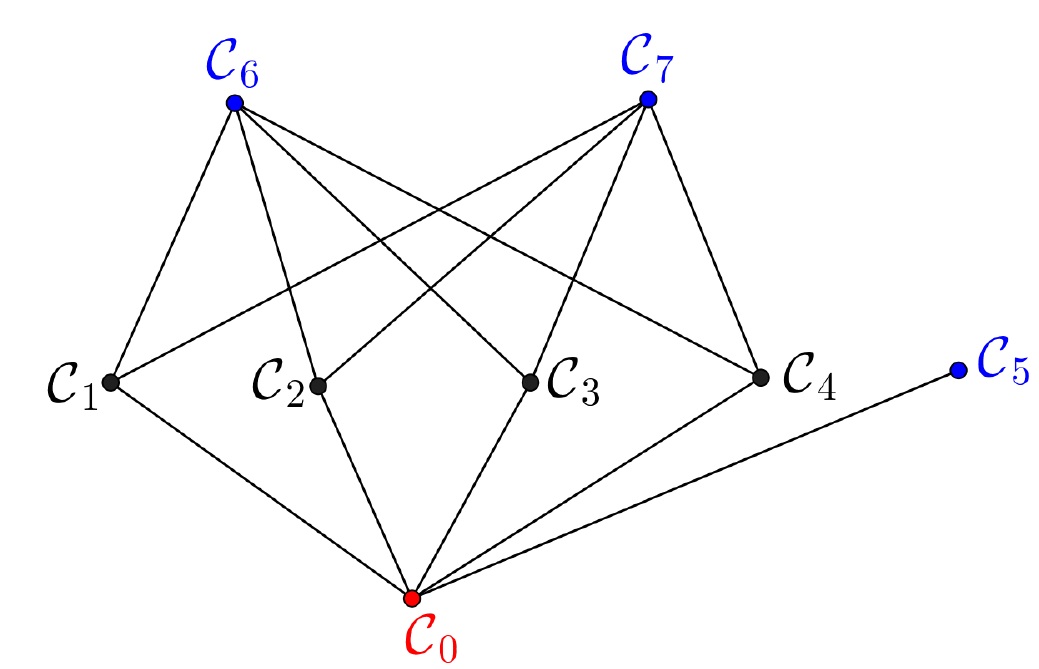}
		\caption{Hasse diagram of the code $\mathcal{C}$.}
		\label{fig:1}
	\end{figure}
\end{ejem}

\textbf{Remark.} If $\mathcal{C}_1$ and $\mathcal{C}_2$ are cosets of $\mathcal{C}$ with $\mathcal{C}_1 \prec \mathcal{C}_2$, then $\mathcal{C}_1$ is called a \emph{descendant} of $\mathcal{C}_2$, and
$\mathcal{C}_2$ is an \emph{ancestor} of $\mathcal{C}_1$. Note that the coset $\mathcal{C}_0=\boldsymbol{0}+\mathcal{C}$ is always a descendant of any coset of the linear code $\mathcal{C}$, in other words $\mathcal{C}_0$ is minimal in the poset $(\mathfrak{cl}(\mathcal{C}),\preceq)$. Always, this vertex will be represented with a red dot in the figures. A coset of $\mathcal{C}$ is called an \emph{orphan} whenever it has no parents, that is, it is a maximal element in the poset $(\mathfrak{cl}(\mathcal{C}),\preceq)$. These vertices will be denoted, in the figures, with a blue dot, see Figure \ref{fig:1}.

\begin{ejem}
Let $\mathcal{C}$ be a binary linear code with parameters $[6,2,4]$ with generator matrix
$$G= \begin{pmatrix}
		1 & 1 & 1 & 1 & 0 & 0\\
		0 & 0 & 1 & 1 & 1 & 1\\
	\end{pmatrix}.$$
	It can be verified that $\mathcal{C}$ has 16 cosets, its graph $\Gamma(\mathcal{C})$ is given in the Figure \ref{ejem_grafo4} and the cosets $\mathcal{C}_7= +\mathcal{C}$, $\mathcal{C}_{14}=+ \mathcal{C}$ and $\mathcal{C}_{15}=+ \mathcal{C}$ are its orphans.
	
	\begin{figure}[H]
		\centering
		\includegraphics[scale=0.2]{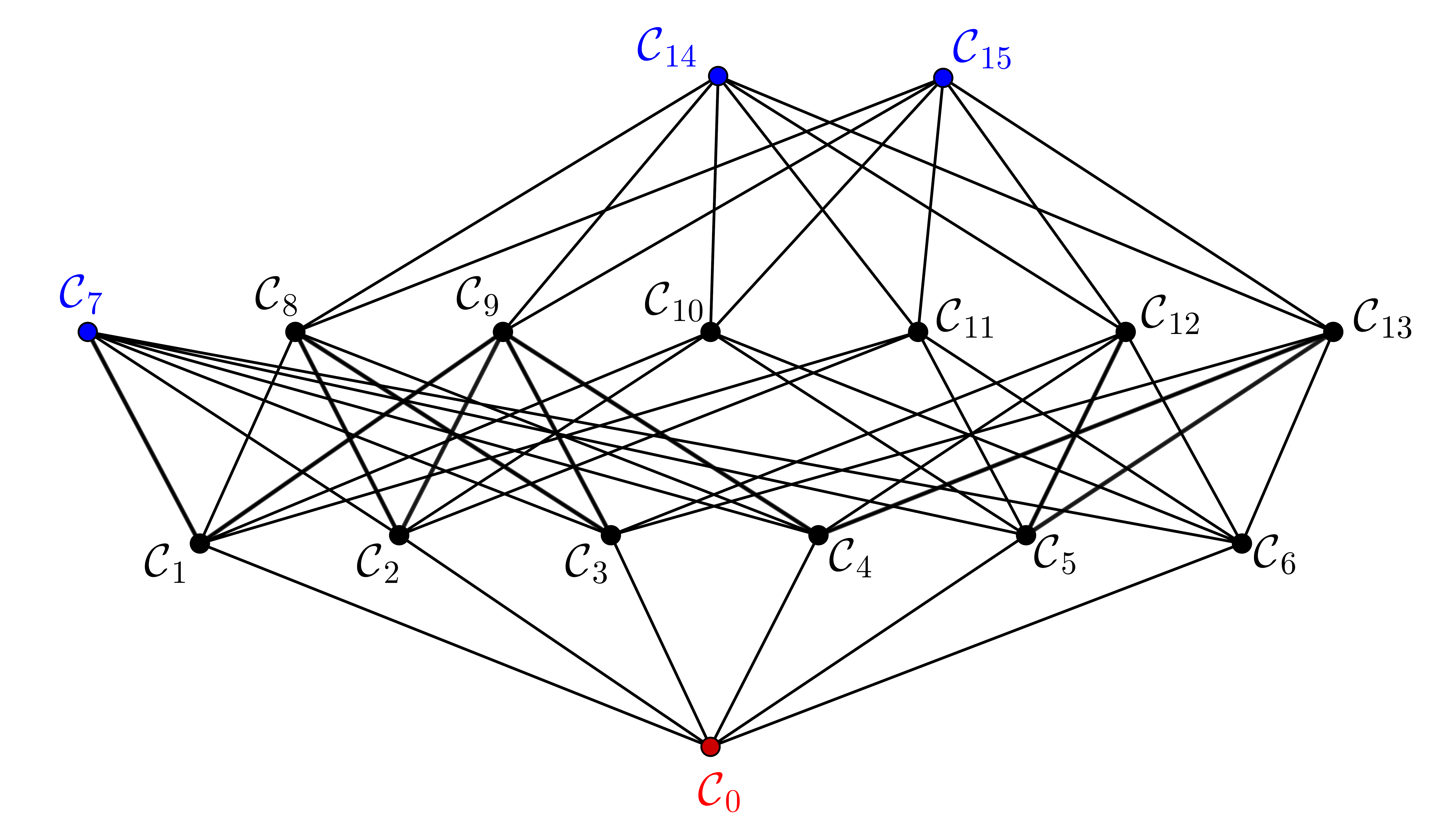}
		\caption{Hasse diagram of the code $\mathcal{C}$.}
		\label{ejem_grafo4}
\end{figure}
\end{ejem}

Let us recall other definitions of graph theory, the reader can find them in \cite{chartrand2013first,B98}. Given a
graph $\Gamma$, a \emph{$u-v$ walk} $W$ in $\Gamma$ is an alternating sequence of vertices and edges
$v_0e_1v_1e_2\ldots v_{n-1}e_nv_n$ beginning with $v_0=u$ and ending at $v_n=v$, in which each
edge is incident with the two vertices immediately preceding and following it, i.e., $e_i=v_{i-1}v_i, ~1\leq i\leq n$.
This walk joins $v_0$ and $v_n$, and may also be denoted $W=v_0v_1v_2\cdots v_n$. The number of edges in
a walk (including multiple occurrences of an edge) is called the \emph{length} of the walk. We say that $W$ is \emph{closed} if $v_0=v_n$ and
is \emph{open} otherwise. It is a \emph{trail} if all the edges are distinct, and a \emph{path} if all the vertices (and thus necessarily
all the edges) do not repeat. Hence every path is a trail, but the converse is not true. If the walk is closed, then it
is a \emph{cycle} provided its $n$ vertices are distinct and $n\geq 3$. Hence a cycle begins and ends at the same vertex but
repeats no edges. A cycle of length $n$ is called a \emph{$n$-cycle}. As is usual, we denote by $C_n$ the graph consisting
of a cycle with $n$ vertices and by $P_n$ a path with $n$ vertices; $C_3$ is often called a \emph{triangle}. A cycle of odd length
is called an \emph{odd cycle}; while, otherwise it is called an \emph{even cycle}. A graph without triangles is
called a \emph{triangle-free graph}.

A graph $\Gamma$ is \emph{connected} if every two vertices of $\Gamma$ are joined by a path. A graph $\Gamma$ that is not connected
is called \emph{disconnected}. A maximal connected subgraph of $\Gamma$ is called a \emph{connected component} or simply a component of
$\Gamma$. Thus, a disconnected graph has at least two components. Let $\Gamma=(V,E)$ be a connected graph and $u, v\in V$ two vertices.
The \emph{distance} between $u$ and $v$ is the smallest length of any $u-v$ path in $G$ and is denoted by $d_\Gamma(u,v)$. The
greatest distance between any two vertices of a connected graph $\Gamma$ is called the \emph{diameter} of $\Gamma$ and is denoted by
$\diam(\Gamma)$.

 A \emph{coloring} of a graph $\Gamma$, is an assignment of colors (elements of some set) to the vertices of $\Gamma$, one
color to each vertex, in such a way that adjacent vertices have distinct colors. The smallest number of colors in any coloring
of a graph $\Gamma$ is called the \emph{chromatic number} of $\Gamma$ and is denoted by $\chi(\Gamma)$. A graph $\Gamma$ such
that $\chi(\Gamma)=2$ is a \emph{bipartite graph}, i.e., the set of vertices of $\Gamma$ can be decomposed into two disjoint sets
such that no two graph vertices within the same set are adjacent.

\section{Main results}\label{sec3}
In this section we prove the main results about some properties of the graph $\Gamma(\mathcal{C})$, for a given binary linear code $\mathcal{C}$.

\begin{teor}
Let $\mathcal{C}$ be a binary linear code. Then $\Gamma(\mathcal{C})$ is connected.
\end{teor}

\begin{proof}
	Let $\mathcal{C}_k$ and $\mathcal{C}_l$ be two  cosets of $\mathcal{C}$, such that $\wt(\mathcal{C}_k)\leq \wt(\mathcal{C}_l)$, with $\boldsymbol{x}_k$ and $\boldsymbol{x}_l$ their coset leaders, respectively. Then we study the following cases.
	\begin{enumerate}[(i)]
		\item Suppose that $\mathcal{C}_k$ is a descendant of $\mathcal{C}_l$. Take $A_1=\supp(\boldsymbol{x}_l)\setminus \supp(\boldsymbol{x}_k)$, $t=|A_1|$, $i_1=\min A_1$, $\boldsymbol{x}_{i_1}=\boldsymbol{x}_l+\boldsymbol{e}_{i_1}$ and for $2\leq s\leq t$
		\begin{align*}
			A_s&=A_{s-1}\setminus \{i_{s-1}\},\\
			i_s&=\min{A_s}, \\
			\boldsymbol{x}_{i_s}&=\boldsymbol{x}_{i_{s-1}}+\boldsymbol{e}_{i_s}.
		\end{align*}
		Additionally, consider the cosets $\mathcal{C}_{i_s}=\boldsymbol{x}_{i_s}+\mathcal{C}$, for $1\leq s\leq t$. By \cite[Corollary 11.7.7]{Huffman2003},  $\boldsymbol{x}_{i_1}$ is a coset leader of a child of $\mathcal{C}_l$ and also $\mathcal{C}_{i_s}$ is a child of $\mathcal{C}_{i_{s-1}}$. Thus  $\mathcal{C}_l\mathcal{C}_{i_1} \cdots \mathcal{C}_{i_t}$ is a $\mathcal{C}_l-\mathcal{C}_k$ path, where $\mathcal{C}_{i_t}=\mathcal{C}_k$. Observe that this path is the shortest between these two vertices, but it is not necessarily unique. As a consequence we get that $d_{\Gamma}(\mathcal{C}_{k},\mathcal{C}_{l})=t$.
		
		\item Assume that $\mathcal{C}_k$ is not a descendant of $\mathcal{C}_l$. Take $$\Omega=\{\boldsymbol{y}+\mathcal{C}: \boldsymbol{y}+\mathcal{C}\prec \mathcal{C}_k \text{ and } \boldsymbol{y}+\mathcal{C}\prec \mathcal{C}_l\},$$ that is $\Omega$ is the set of the common descendants of $\mathcal{C}_k$ and $\mathcal{C}_l$. Since that $\mathcal{C}_0=\boldsymbol{0}+\mathcal{C}$ always is a descendant of any coset of $\mathcal{C}$, then $\Omega\neq \emptyset$.  We take $\mathcal{C}''$ such that $\wt(\mathcal{C}'')=\max\{\wt(\mathcal{C}'): \mathcal{C}'\in \Omega\}$.
		
Since $\mathcal{C}''$ is a descendant of $\mathcal{C}_k$, by (i) there exists a  $\mathcal{C}_k-\mathcal{C}''$ path in $\Gamma(\mathcal{C})$, we call it $\mathcal{C}_k\cdots\mathcal{C}''$. Similarly, there is a $\mathcal{C}''-\mathcal{C}_l$ path: $\mathcal{C}''\mathcal{C}''_1\cdots \mathcal{C}_l$. Thus by the choose of $\mathcal{C}''$ we find a path from $\mathcal{C}_k$ to $\mathcal{C}_l$ in $\Gamma(\mathcal{C})$:  $\mathcal{C}_k\cdots \mathcal{C}''\mathcal{C}''_1\cdots \mathcal{C}_l$.
	\end{enumerate}
\end{proof}

For a coset $\mathcal{C}_k$ of a binary linear $\mathcal{C}$, with $CL(\mathcal{C}_k)$ we denote a complete set of the coset leaders of $\mathcal{C}_k$. Now, from the previous proof, we get the following result.

\begin{cor}\label{cor2.1}
	Let $\mathcal{C}$ be a binary linear code, $\mathcal{C}_k$ and $\mathcal{C}_l$ cosets of $\mathcal{C}$ such that $\boldsymbol{x}_k\in CL(\mathcal{C}_k)$, $\boldsymbol{x}_l\in CL(\mathcal{C}_l)$. Then
	$$
	d_{\Gamma}(\mathcal{C}_k,\mathcal{C}_l)=
	\begin{cases}
		d(\boldsymbol{x}_k,\boldsymbol{x}_l), & \text{if  $\mathcal{C}_k\prec \mathcal{C}_l$ and $\boldsymbol{x}_k\prec \boldsymbol{x}_l$},\\
		d(\boldsymbol{x}_k,\boldsymbol{x}'')+d(\boldsymbol{x}'',\boldsymbol{x}_l), & \text{if  $\mathcal{C}_k\not\prec \mathcal{C}_l$ and $\boldsymbol{x}''\in CL(\mathcal{C}'')$, $\mathcal{C}''\in \Omega$}\\
			& \text{such that $\wt(\mathcal{C}'')=\max\{\wt(\mathcal{C}'): \mathcal{C}'\in \Omega\}$.}
	\end{cases}
	$$
\end{cor}

Since $\mathcal{C}_0$ is a descendant of any coset of $\mathcal{C}$, if $\mathcal{C}_l$ is an orphan such that $\rho(\mathcal{C})=\wt(\mathcal{C}_l)$, that is $\mathcal{C}_l$ is a coset with the largest Hamming weight, see  \cite[Theorem 11.1.2]{Huffman2003} then we get the following result.

\begin{cor}
	Let  $\mathcal{C}$ be a binary linear code. Then  $\diam(\Gamma(\mathcal{C}))\geq \rho(\mathcal{C})$.
\end{cor}

The problem of deciding whether a given graph $\Gamma = (V, E)$ contains a
$k$-cycle is among the most natural and easily stated algorithmic graph
problems. In particular, the \emph{triangle finding problem} is the problem of determining whether a graph is triangle-free or not. It is possible to determine if a given graph with $m$ edges is triangle-free in time $O(m^{1.41})$, see \cite[Thm. 3.5]{AYZ97}. In the following result, we solve this problem to the family of  graphs $\Gamma(\mathcal{C})$.

\begin{teor}\label{triangle_free}
Let $\mathcal{C}$ be a binary linear code. Then $\Gamma(\mathcal{C})$ is a triangle-free graph.
\end{teor}
\begin{proof}
Assume that there is a triangle in $\Gamma(\mathcal{C})$ with vertices $\mathcal{C}_k$, $\mathcal{C}_q$ and $\mathcal{C}_m$. Suppose that $\mathcal{C}_k\mathcal{C}_q \in E_{\mathcal{C}}$. Without lost of generality, suppose that $\mathcal{C}_k$ is a child of $\mathcal{C}_q$. Since $\mathcal{C}_k\mathcal{C}_m\in E_{\mathcal{C}}$, then either $\mathcal{C}_k$ is a child of $\mathcal{C}_m$ or $\mathcal{C}_m$ is a child of $\mathcal{C}_k$.

If  $\mathcal{C}_k$ is a child of $\mathcal{C}_m$, then  $\wt(\mathcal{C}_m)=\wt(\mathcal{C}_q)$ and then $\mathcal{C}_q=\mathcal{C}_m$, which is a contradiction.	On the other hand, if  $\mathcal{C}_m$ is a child of $\mathcal{C}_k$, then $\wt(\mathcal{C}_m)=\wt(\mathcal{C}_k)-1$, furthermore $\wt(\mathcal{C}_k)=\wt(\mathcal{C}_q)-1$, thus $\wt(\mathcal{C}_m)=\wt(\mathcal{C}_q)-2$, thus there no exists an edge between $\mathcal{C}_m$ and $\mathcal{C}_q$, so again we obtain a contradiction. Consequently, $\mathcal{C}_k$ and  $\mathcal{C}_m$ are not adjacent vertices in  $\Gamma(\mathcal{C})$, and therefore it is a triangle-free graph.
\end{proof}

The study of the chromatic number of triangle‐free graphs is a classic topic, which has been studied extensively for several researchers and from many perspectives (for instance algebraic, probabilistic, quantitative Ramsey theory and algorithmic), see \cite{davies2020coloring} for more references about this problem.  In fact, triangle-free graphs can have large chromatic numbers, see \cite[Thm. 10.10]{chartrand2013first}. The next theorem demonstrates that if $\mathcal{C}$ is a proper vectorial subspace, then $\chi(\Gamma(\mathcal{C}))=2$.

\begin{teor}
Let $\mathcal{C}\subset \mathbb{F}_2^n$ be a binary linear code. Then $\Gamma(\mathcal{C})$ is a  bipartite graph.
\end{teor}

\begin{proof}
Since $\mathcal{C}\subset \mathbb{F}_2^n$, it is clear that  $|V_{\mathcal{C}}|\geq 2$. Consider $U$ and $V$, the set of cosets of $\mathcal{C}$ of odd weight and even weight, respectively. It is easy to prove that $V_{\mathcal{C}}=U\cup V$ and $U\cap V=\emptyset$. Now, by the definition of $\Gamma(\mathcal{C})$, the vertices of even weight are adjacent uniquely with vertices of odd weight and vice versa. Therefore, any edge in $E_{\mathcal{C}}$ connect nodes from $V_1$ with elements in $V_2$, but not connect vertices in the same set of vertices and the result follows.
\end{proof}
It is known that a nontrivial graph is bipartite  if and only if  it does not contains no odd cycles, see \cite[Theorem 1.12]{chartrand2013first}. Thus the last result generalize Theorem \ref{triangle_free}.

For $n\in \mathbb{Z}^+$, let $S_n$ denote the \emph{symmetric group on $n$ letters}. For   $\pi\in S_n$, a permutation on the coordinates of $\mathbb{F}_2^n$ is a function on $\mathbb{F}_2^n$ defined by $\mathcal{P}(u_1\ldots u_n)=u_{\pi(1)}\ldots u_{\pi(n)}$.  We remember that two binary linear codes $\mathcal{C}$ and $\mathcal{D}$ are \emph{permutation equivalent} if there is a permutation on the coordinates of $\mathbb{F}_2^n$ which sends $\mathcal{C}$ to $\mathcal{D}$, i.e. there exists a permutation $\pi \in S_n$ such that for any codeword $\boldsymbol{v}\in \mathcal{D}$ there exists a codeword $\boldsymbol{u}\in \mathcal{C}$ such that $\boldsymbol{v}=\mathcal{P}(\boldsymbol{u})$. Since that a permutation on the coordinates of $\mathbb{F}_2^n$ preserves the weight of any word, it can be proved that two permutation equivalent binary linear codes have the same parameters, thus they can be considered as ``the same code''. Deciding whether two codes are permutation equivalent codes is known as the \emph{Code equivalence problem}. This problem has been extensively studied in the last decades, see \cite{PR97,sendrier2000finding} and the references therein.

Recall that two graphs are \emph{isomorphic} if there is a correspondence between their vertex sets that preserves adjacency, i.e.,
$\Gamma=(V,E)$ is isomorphic to $\Gamma'=(V',E')$, denoted by $\Gamma\simeq \Gamma'$, if there is a a bijection $\varphi$ from
$V$ to $V'$ such that $xy\in E$ if and only if $\varphi(x)\varphi(y)\in E'$.  The \emph{graph isomorphism problem} requests to determine whether two given graphs are isomorphic. It is well known that deciding if two graphs are isomorphic is an algorithmic problem that has been widely studied, see \cite{Babai2018}.

\begin{teor}\label{grafos_isomorfos}
Let $\mathcal{C}$ and $\mathcal{D}$ two linear binary codes. If $\mathcal{C}$ and $\mathcal{D}$ are permutation equivalent codes, then $\Gamma(\mathcal{C})\simeq \Gamma(\mathcal{D})$.
\end{teor}
\begin{proof}
By the assumption there exists  $\mathcal{P}$ a permutation on the coordinates of $\mathbb{F}_2^n$ such that $\mathcal{D}=\{\mathcal{P}(\boldsymbol{u}):\boldsymbol{u}\in \mathcal{C}\}$. Take $\boldsymbol{u}$ a coset leader of $\boldsymbol{u}+\mathcal{C}$. Since $\mathcal{P}$ preserves Hamming distance, we have that for any $\boldsymbol{c}\in \mathcal{C}$, $\wt(\boldsymbol{u}+\boldsymbol{c})=\wt(\mathcal{P}(\boldsymbol{u})+\boldsymbol{v})$, where $\boldsymbol{v}=\mathcal{P}(\boldsymbol{c})$. Then $\mathcal{P}(\boldsymbol{u})$ is a leader of its coset in $\mathcal{D}$, that is $\mathcal{P}(\boldsymbol{u})+\mathcal{D}$ is a vertex of $\Gamma(\mathcal{D})$.

Let $\mathcal{C}_1=\boldsymbol{u}+\mathcal{C}$ and $\mathcal{C}_2=\boldsymbol{u}'+\mathcal{C}$ be two cosets of $\mathcal{C}$.  Now, if $\mathcal{C}_1\mathcal{C}_2$ is an edge of $\Gamma(\mathcal{C})$, then $\mathcal{C}_1\prec \mathcal{C}_2$. So  there is a coset leader $\boldsymbol{u}''$ of $\mathcal{C}_2$ such that $\boldsymbol{u}\prec \boldsymbol{u}''$; i.e. $\supp(\boldsymbol{u})\subset \supp(\boldsymbol{u}'')$. This last statement implies that $\supp(\mathcal{P}(\boldsymbol{u}))\subset \supp(\mathcal{P}(\boldsymbol{u}''))$ or equivalently we have that $\mathcal{P}(\boldsymbol{u})+\mathcal{D}$ and $\mathcal{P}(\boldsymbol{u}'')+\mathcal{D}=\mathcal{P}(\boldsymbol{u}')+\mathcal{D}$ are adjacent in $\Gamma(\mathcal{D})$. So the function $\varphi: V_{\mathcal{C}}\rightarrow V_{\mathcal{D}}$ given by $\varphi(\boldsymbol{u}+\mathcal{C})=\mathcal{P}(\boldsymbol{u})+\mathcal{D}$ is a graph isomorphism.
\end{proof}

\begin{ejem}
The reciprocal of Theorem \ref{grafos_isomorfos}  is false. In fact, consider the binary linear codes $\mathcal{C}$  and $\mathcal{D}$ with generator matrices $$M=\begin{pmatrix}
     	1 & 0 & 0 & 0 & 1 & 1\\
     	0 & 1 & 0 & 1 & 0 & 1\\
     	0 & 0 & 1 & 1 & 1 & 0\\
     \end{pmatrix} \text{ and }  N= \begin{pmatrix}
		1 & 0 & 0 & 0 & 0 & 1 & 1 & 0\\
		0 & 1 & 0 & 0 & 0 & 1 & 0 & 1\\
		0 & 0 & 1 & 0 & 0 & 0 & 1 & 0\\
		0 & 0 & 0 & 1 & 0 & 0 & 0 & 1\\
		0 & 0 & 0 & 0 & 1 & 1 & 1 & 1\\
	\end{pmatrix},$$
	respectively.  It can be verified that $d(\mathcal{C})=3$ and $d(\mathcal{D})=2$, fact that implies that these codes are non permutation equivalent codes. However, their Hasse diagrams are isomorphic, see Figure \ref{fig_grafos_isomorfos}.

\begin{figure}[H]
	\centering
	\hspace*{2.5cm}\includegraphics[scale=0.25]{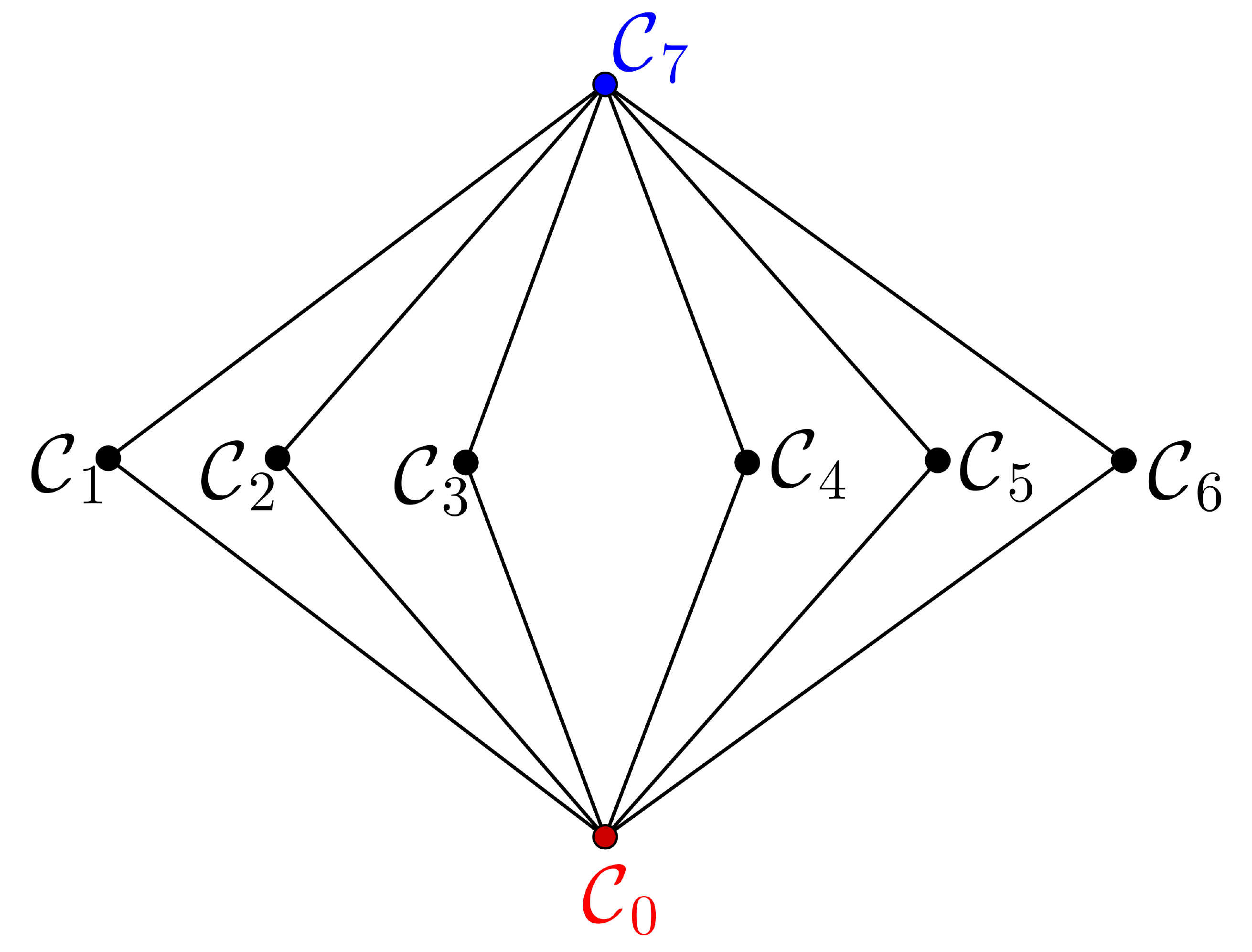}
	\caption{Hasse diagram of $\mathcal{C}$ and $\mathcal{D}$
.}
	\label{fig_grafos_isomorfos}
\end{figure}
\end{ejem}

Theorem \ref{grafos_isomorfos} and the above example motivate us to ask when two non permutation equivalent binary codes $\mathcal{C}$ and $\mathcal{D}$ satisfy that $\Gamma(\mathcal{C})\simeq \Gamma(\mathcal{D})$. More  precisely:

\begin{ques}
What conditions must satisfy two non permutation equivalent binary codes to guarantee that their Hasse diagrams are structurally identical?
\end{ques}

Now that we know that $\Gamma(\mathcal{C})$ is a connected bipartite graph, we can give a particular answer to the last question. In graph theory, for a positive integer $k$, the \emph{star $S_k$} is the complete bipartite graph $K_{1,k}$. Note that if $\Gamma(\mathcal{C})\simeq S_{2^{n-k}}$, this implies that all the cosets of $\mathcal{C}$ have weight at most 1; that is for all $\boldsymbol{u}\notin \mathcal{C}$ there is $i\in \{1,\ldots,n\}$ such that $\boldsymbol{u}\in \boldsymbol{e}_i+\mathcal{C}$, or equivalently  $\boldsymbol{u}=\boldsymbol{e}_i+\boldsymbol{c}$ for some $\boldsymbol{c}\in \mathcal{C}$, thus $d(\boldsymbol{u},\boldsymbol{c})=1$. It means that  for all $\boldsymbol{u}\in \mathbb{F}_2^n$,  there exists $\boldsymbol{c}\in \mathcal{C}$ such that  $d(\boldsymbol{u},\boldsymbol{c})\leq 1$. Therefore, $\mathbb{F}_2^n=\bigcup_{\boldsymbol{c}\in \mathcal{C}}S(\boldsymbol{c},1)$; namely $\rho(\mathcal{C})=1$. It is easy to see that if $\rho(\mathcal{C})=1$, then $\mathfrak{cl}(\mathcal{C})=\{\boldsymbol{e}_i+\mathcal{C}: \text{ for some } 1\leq i\leq n\}$. Consequently, the Hasse diagram  $\Gamma(\mathcal{C})$ is isomorphic to a graph as that one given in Figure \ref{fig:estrellainfinita}.

\begin{figure}[H]
		\centering
		\hspace*{2.5cm}\includegraphics[scale=0.4]{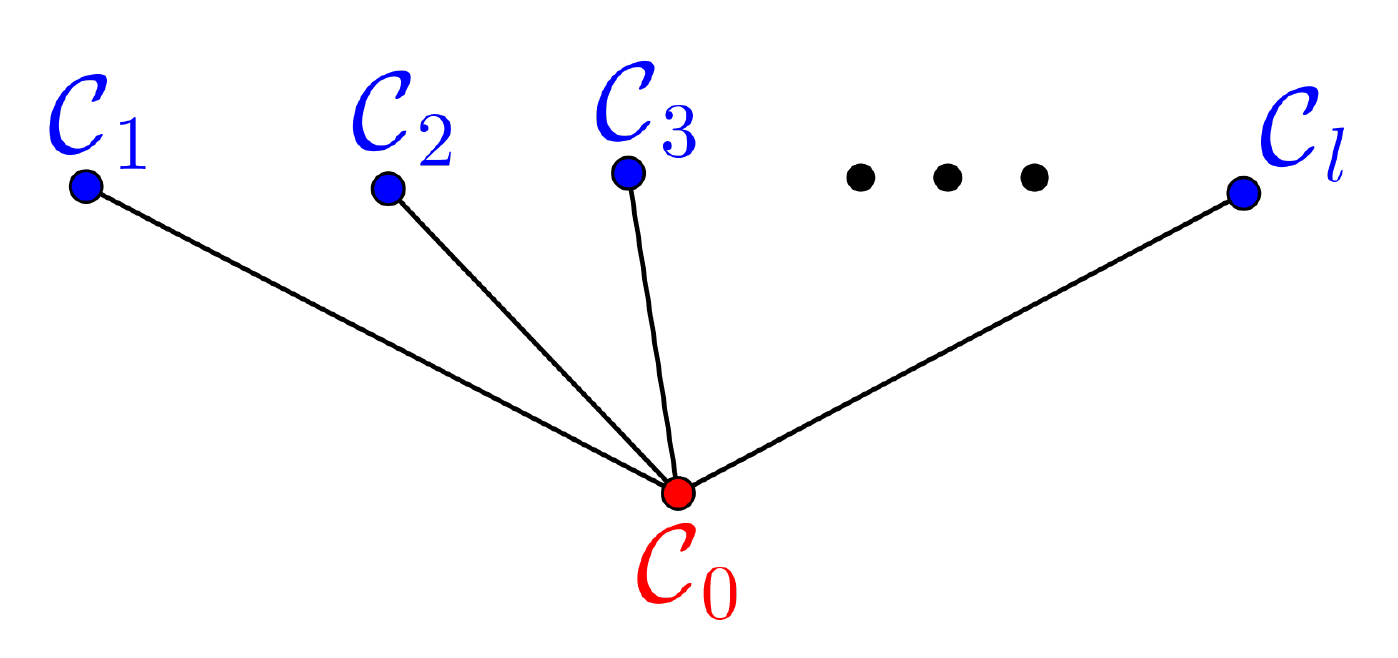}
		\caption{$\Gamma(\mathcal{C})$ with $\rho(\mathcal{C})=1$.}
		\label{fig:estrellainfinita}
	\end{figure}
	
Actually, in the last paragraph we have proven the following result.
\begin{teor} Let $\mathcal{C}$ a binary linear code. Then
$\Gamma(\mathcal{C})\simeq S_{2^{n-k}}$ the star graph if and only if $\rho(\mathcal{C})=1$.

\end{teor}
Binary codes with covering radius 1 has been studied in the literature, see \cite{habsieger1997} and the references quoted therein.

\end{document}